\documentclass{article}

\usepackage[utf8]{inputenc}
\usepackage{amsmath, amssymb, amsthm}
\usepackage{mathtools}
\usepackage{fullpage}
\usepackage{comment}

\usepackage[
backend=biber,
style=alphabetic,
citestyle=alphabetic,
maxalphanames=4,
maxcitenames=99,
mincitenames=98,
maxbibnames=99,
giveninits=true,
]{biblatex}
\usepackage{todonotes}
\presetkeys{todonotes}{}{}

\usepackage[ruled]{algorithm2e}
\usepackage{setspace}
\usepackage{xspace}
\usepackage{array}
\usepackage[inline]{enumitem}
\setlist[itemize,enumerate]{
topsep=0.4em,
itemsep=0.2em,
parsep=0em}

\setlist[itemize,enumerate,2]{
topsep=0.2em
}

\graphicspath{{graphics/}}

\usepackage[capitalize,nameinlink]{cleveref}
\crefname{@theorem}{theorem}{theorem}

\usepackage{setspace}

\usepackage{mathtools}
\usepackage{complexity}

\newcommand{\Mrel}{M_{\mathrm{relax}}}
\newcommand{\Mext}{M_{\mathrm{ext}}}
\newcommand{\Jext}{\mathcal{J}_{\mathrm{ext}}}

\DeclareMathOperator{\argmax}{argmax}

\DeclareMathOperator{\rank}{rank}

\newtheorem{theorem}{Theorem}
\newtheorem{definition}[theorem]{Definition}
\newtheorem{lemma}[theorem]{Lemma}

\title{A Simple Combinatorial Algorithm for Robust Matroid Center}
\author{Georg Anegg\thanks{ETH Zurich. Email: ganegg@ethz.ch. Research supported in part by Swiss National Science Foundation grant number 200021\_184622.} \and Laura Vargas Koch\thanks{ETH Zurich. Email: lvargas@ethz.ch} \and
Rico Zenklusen\thanks{ETH Zurich. Email: ricoz@ethz.ch. Research supported in part by Swiss National Science Foundation grant number 200021\_184622. This project has received funding from the European Research Council (ERC) under the European Union's Horizon 2020 research and innovation programme (grant agreement No 817750).}}
\date{}

\begin{document}

\maketitle

\begin{abstract}
\small\baselineskip=9pt
Recent progress on robust clustering led to constant-factor approximations for Robust Matroid Center.
After a first combinatorial $7$-approximation that is based on a matroid intersection approach, two tight LP-based $3$-approximations were discovered, both relying on the Ellipsoid Method.
In this paper, we show how a carefully designed, yet very simple, greedy selection algorithm gives a $5$-approximation.
An important ingredient of our approach is a well-chosen use of Rado matroids.
	This enables us to capture with a single matroid a relaxed version of the original matroid, which, as we show, is amenable to straightforward greedy selections.
\end{abstract}
 
\begin{tikzpicture}[overlay, remember picture, shift = {(current page.south west)}]
\begin{scope}[shift={(8.1,1.8)}]
\def\hd{2.5}
\node at (-2.15*\hd,0) {\includegraphics[height=0.7cm]{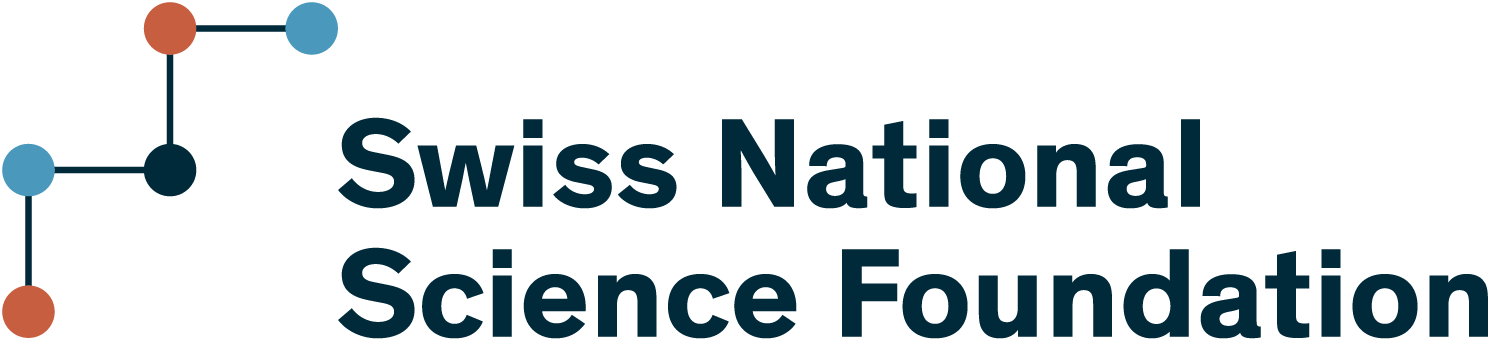}};
\node at (-\hd,0) {\includegraphics[height=1.0cm]{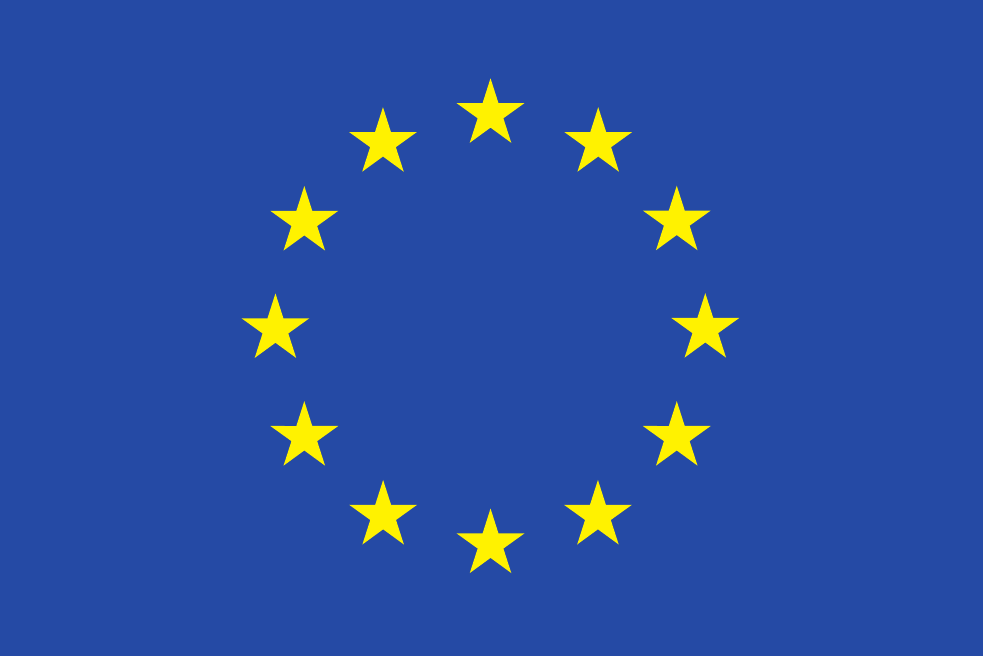}};
\node at (-0.2*\hd,0) {\includegraphics[height=1.2cm]{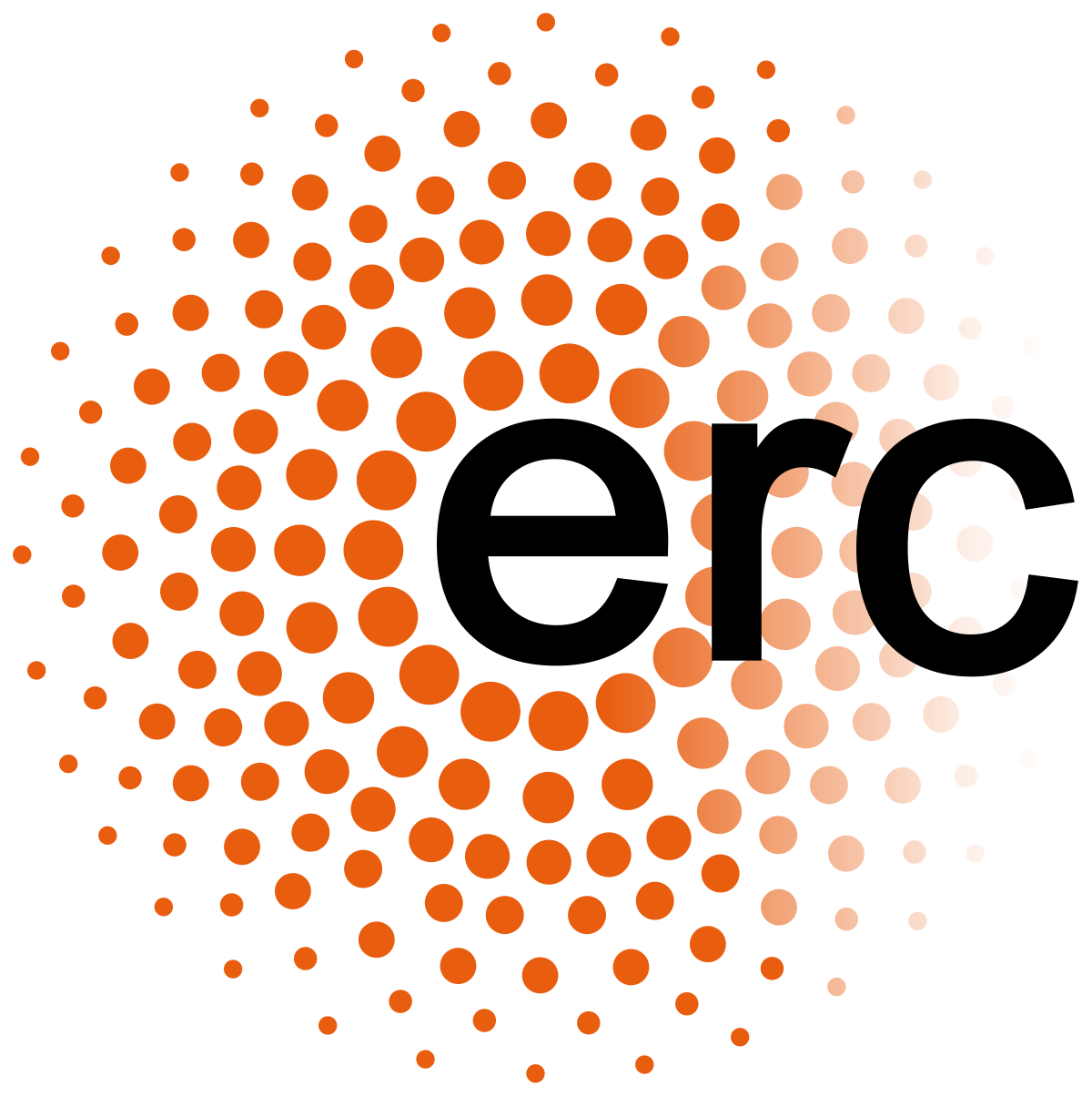}};
\end{scope}
\end{tikzpicture}

\section{Introduction}

Robust clustering (also referred to as clustering with outliers) is a widely applicable family of clustering problems, which has been studied in a variety of settings, see, e.g., \cite{CKMN01, CLLW16, HPST19, CN19} for more information on Robust $k$-Center variants.
Robust clustering variants, like Robust $k$-Center, do not require to cluster all points, but only at least a given number $m$ of points, i.e., one can declare all except $m$ many points as \emph{outliers} and not cluster them.

For the classical $k$-Center problem, tight $2$-approximations are known through two different fast combinatorial algorithms by \textcite{G85,HS86}, both of which are natural and straightforward.
While Robust $k$-Center also admits (LP-based) $2$-approximations (see \textcite{chakrabarty_2016_non-uniform}), the best combinatorial algorithm is an elegant greedy algorithm by \textcite{CKMN01} which gives a $3$-approximation.
For the Robust Matroid Center problem, which we define below, the situation is somewhat similar, but based on more involved algorithms.
(In the paper,
we use the notation $B(u,r)\coloneqq \{v\in X \mid d(u,v)\leq r\}$ for a metric space $(X, d)$ and $u\in X$ as well as the shorthand $B(U,r)\coloneqq \bigcup_{u\in U} B(u,r)$ for $U\subseteq X$. Moreover, $w(U) \coloneqq \sum_{u \in U} w(u)$ for any $w\colon X\to \mathbb{R}_{\geq 0}$.)
\begin{definition}[Robust Matroid Center]
Let $(X,d)$ be a finite metric space with weights $w\colon X\to \mathbb{R}_{\geq 0}$, let $M=(X,\mathcal{I})$ be a matroid, and let $m\in \mathbb{R}_{\geq 0}$ be a lower bound on the weight to be clustered.\footnote{We recall that a \emph{matroid} $M=(X,\mathcal{I})$ consist of a finite ground set $X$ and a nonempty family $\mathcal{I}\subseteq 2^X$ of subsets of $X$ (called \emph{independent sets}) fulfilling
\begin{enumerate*}[label=(\roman*)]
\item if $I\in \mathcal{I}$ and $J\subseteq I$ then $J\in \mathcal{I}$, and
\item if $I,J\in \mathcal{I}$ with $|J|>|I|$ then $\exists e\in J\setminus I$ with $I\cup\{e\}\in \mathcal{I}$
\end{enumerate*}.
We make the standard assumption that a matroid $M=(X,\mathcal{I})$ is given through an independence oracle, which, when queried with a set $S\subseteq X$, returns whether $S\subseteq \mathcal{I}$.
}
The Robust Matroid Center problem asks to determine the smallest possible $r\in \mathbb{R}_{\geq 0}$ together with a set $C\subseteq X$ such that
\begin{enumerate}[nosep]
\item $C\in \mathcal{I}$, and
\item $w(B(C,r)) \geq m$.\footnote{
Observe that when $M$ is the uniform matroid of rank $k$, i.e., when $\mathcal{I}=\{C\subseteq X\colon |C|\leq k\}$, we recover the Robust $k$-Center problem, and when additionally we also have $m=|X|$, we recover the classic $k$-Center problem.
}
\end{enumerate}
\end{definition}
The first constant-factor approximation for Robust Matroid Center is due to \textcite{CLLW16}, who presented a combinatorial $7$-approximation through a reduction to matroid intersection.
Later on, tight $3$-approximations have been obtained by \textcite{HPST19,CN19}, both being LP-based and using the Ellipsoid Method.
The tightness of these $3$-approximations follows from a result by \textcite{HS86}, who showed that no $(3-\epsilon)$-approximation exists unless $\P=\NP$.

Contrary to Robust $k$-Center, no constant-factor approximation through a simple greedy selection procedure was known for Robust Matroid Center prior to our work.
The goal of this work is to address this gap.

The above-mentioned greedy algorithm of \cite{CKMN01} for Robust $k$-Center, or other canonical greedy approaches, do not easily extend to the matroid setting.
The key issue when trying to pick centers greedily in the Robust Matroid Center setting is that the selection of a center earlier on may restrict future candidates in important ways.
To get around this, \cite{CLLW16} adapt the algorithm of \cite{CKMN01} by not picking centers immediately but using the greedy procedure to first carefully construct an auxiliary problem, which can be solved via matroid intersection in a second step.
Solving this auxiliary problem amounts to optimizing globally where to open centers.\footnote{It has later been observed that important steps of \cite{CLLW16} do not intrinsically rely on the matroid structure and can be extended to any down-closed families of feasible facilities \cite{AVZ22}.}

\paragraph{Our contribution.}
We improve on the best combinatorial algorithm by presenting a $5$-approxi\-mate greedy selection algorithm that is both simpler and stronger than the combinatorial $7$-approximation of \cite{CLLW16}.

\begin{theorem}
\label{thm:5approx}
There is a combinatorial $5$-approximation algorithm for Robust Matroid Center.
\end{theorem}
To get around the selection problem discussed earlier, we relax the matroid condition and allow each selected center to be represented by any nearby point.
This can be captured by an auxiliary matroid known as Rado matroid. A given set of centers is independent in the Rado matroid if there is an appropriate independent set (in the original matroid) of representatives close-by to the given centers.
We then apply a natural greedy algorithm with respect to this relaxed matroid, and show that this addresses the above-mentioned problem of selecting centers greedily in a canonical way.
Hence, instead of greedily opening centers right away, our algorithm greedily selects small areas in which centers will (and can) be opened.
We exploit the matroid properties of our relaxed matroid in the analysis by showing that it allows for continuously transforming an optimal solution (say of radius $r$) to our solution (of radius $5r$), center by center.
We obtain coverage guarantees at every step during this transformation, which, eventually, imply that our solution covers at least as much as an optimal one and thus fulfills the coverage requirement.

In contrast, classical greedy approaches attempt to fix an exact location at every step, which leads to the above-explained issues.
Also, contrary to all prior algorithms for Robust Matroid Center, we obtain during the execution of our greedy procedure immediate guarantees on point sets that will be covered.
Due to this, we only need to run our greedy procedure for rank many iterations whereas prior algorithms first partition the entire input space into a possibly large number of parts.

The bottleneck of the runtime of our algorithm is determined by $|X|$ calls to an independence oracle of a Rado matroid induced by $M$. A generic method to solve this problem is by applying a matroid intersection algorithm, whose runtime might vary significantly depending on the given matroid. 
For a general matroid $M$, the combinatorial matroid intersection algorithm by \textcite{CLSSW19} applied to $M$ and the transversal matroid needed for our Rado matroid gives a runtime of $O(|X| \cdot \rank(M) \log(\rank(M)) \cdot \max \{ \mathcal{T}_{ind}, \mathcal{T}_{|X|,\rank(M)} \})$, where $\mathcal{T}_{ind}$ denotes the runtime of the independence oracle of $M$ and  $\mathcal{T}_{a,b} $ denotes the time to compute a maximum cardinality matching in a bipartite graph with $a$ vertices on one side and $b$ on the other side. Note that this corresponds to the independence oracle of our transversal matroid, with $a = \rank(M)$ and $b=|X|$.
This can be solved through a variety of known techniques, including the $O((\rank(M) + |X|)^{5/2})$ time algorithm by Hopcroft and Karp~\cite{HK73}. 

\section{Algorithm}\label{sec:algorithm}

Throughout this section we consider a fixed instance of the Robust Matroid Center problem.
Hence, we are given a finite metric space $(X,d)$ with weights $w\colon X\to \mathbb{R}_{\geq 0}$, a matroid $M=(X,\mathcal{I})$, and a lower bound $m\in \mathbb{R}_{\geq 0}$ on the total weight to be clustered.
Our algorithm, as other procedures in this context, aims at finding a good clustering given a guess  $r\in \mathbb{R}_{\geq 0}$ for the radius of an optimal solution.
Formally, we obtain the following guarantee.
\begin{theorem}\label{thm:fromRToFiveR}
There is an efficient combinatorial algorithm that, for any $r\in \mathbb{R}_{\geq 0}$, if a solution of radius $r$ exists, then the algorithm returns a solution of radius at most $5r$.
\end{theorem}
Observe first that \Cref{thm:fromRToFiveR} clearly implies \Cref{thm:5approx}.
An easy way to see this is by simply running our procedure for each of the $q=O(|X|^2)$ distinct distances $r_1,\ldots, r_q$ (say with $r_1 < r_2 < \ldots < r_q$) between points in $X$, and return the best solution found.\footnote{When talking about distances between points in $X$, we also consider the zero distance; hence, $r_1=0$.}
As one of these $O(|X|^2)$ distances corresponds to the optimal radius, we are sure to obtain a $5$-approximation.
A more efficient alternative is to use binary search for $r$ over the distances $r_1, \ldots, r_q$ with the goal to find an index $i\in \{1,\ldots, q\}$ such that the following holds (where the first condition is ignored for $i=1$):
\begin{enumerate}[nosep,topsep=0.4em,label=(\roman*)]
\item our algorithm, run with $r=r_{i-1}$, fails to return a solution of radius at most $5 r_{i-1}$, and
\item our algorithm, run with $r=r_i$, returns a solution of radius at most $5 r_i$.
\end{enumerate}
In this case, the solution of radius $5 r_i$ is a $5$-approximation to the instance because the first point implies that the optimal radius must be at least $r_i$ as our procedure did not succeed in returning a solution of radius $5 r_{i-1}$ when run with $r=r_{i-1}$.
Hence, this binary search approach leads to $O(\log|X|)$ calls to our procedure that returns a solution of radius at most $5r$ if a solution of radius $r$ exists.

Hence, in what follows, let $r\in \mathbb{R}_{\geq 0}$ be such that a solution of radius $r$ exists to our Robust Matroid Center instance, and it remains to show how to compute a solution of radius $5r$, thus implying \Cref{thm:fromRToFiveR}.

\smallskip

As mentioned, the algorithm we use to prove \Cref{thm:fromRToFiveR} works on a matroid $\Mrel=(X,\mathcal{J})$ that is a relaxed version of the matroid $M$, i.e., $\mathcal{I}\subseteq \mathcal{J}$, and whose independent sets are defined as follows:
\begin{equation}\label{eq:defMrel}
\mathcal{J} \coloneqq \{J \subseteq X \colon \text{there exists } I\in \mathcal{I} \text{ and a bijection }\phi:J\to I \text{ with } d(\phi(e),e)\leq 2r \; \forall e\in J\}\enspace.
\end{equation}

Indeed, $\Mrel$ is a matroid, namely a so-called \emph{Rado matroid} (derived from $M$). Rado matroids are formally defined as follows.
\begin{definition}[Rado matroid]\label{def:radoMat}
Let $M=(X,\mathcal{I})$ be a matroid, $Y$ be a finite ground set, and, for $y\in Y$, let $X_y\subseteq X$.
Then the \emph{Rado matroid} $\overline{M}=(Y,\mathcal{J})$ induced by $(M,\{X_y\}_{y\in Y})$ is defined by the independent sets
\begin{equation}\label{eq:radoIndep}
\mathcal{J}\coloneqq \{J\subseteq Y \colon \text{there exists $I\in \mathcal{I}$ and a bijection $\phi\colon J\to I$ with $\phi(y)\in X_y \;\forall y\in J$}\}.
\end{equation}
Moreover, for $J\in \mathcal{J}$, a set $I\in \mathcal{I}$ satisfying~\eqref{eq:radoIndep} is called a \emph{system of distinct representatives for $J$}.
\end{definition}
Note that $\mathcal{J}$ as defined in~\eqref{eq:defMrel} is indeed of the form described in~\eqref{eq:radoIndep}, which follows by setting $X_e = B(e,2r)$ for $e\in X$.

A well-known fact about Rado matroids is that, given a set $J \in \mathcal{J}$ (and an independence oracle for $M$), one can efficiently determine a system of distinct representatives for $J$.
(For completeness, we recall this basic, yet for us important, fact in \Cref{sec:distinctRepRado}.)

Note that $\Mrel$ has a natural interpretation as a relaxed version of $M$.
Namely, one can think of an independent set $J\in \mathcal{J}$ as committing to open, for each $e\in J$, a (distinct) center at a distance of at most $2r$ from $e$, and $J\in \mathcal{J}$ guarantees that this is possible such that the opened centers are independent in $M$.  

Our procedure is highlighted below as \Cref{alg:greedy}.\footnote{We denote by $\rank(M)$ the rank of the matroid $M$, i.e., the cardinality of a maximum cardinality independent set in $M$.}
It greedily constructs an independent set $T$ in $\Mrel$ by successively choosing a center $t$ of a ball of radius $r$ that covers the largest number of yet uncovered points and that can be added to $T$ such that the set of chosen centers remains independent in the relaxed matroid. After choosing $t$, all points within distance $3r$ of $t$ are declared as covered.
Finally, a set $R\in \mathcal{I}$ of distinct representatives of $J$ is returned, which exists by definition of $\Mrel$ and, as mentioned, can be found efficiently by well-known results on Rado matroids.

We make no assumptions on which maximizer is used in the for-loop in case of ties.

\begin{algorithm}[ht]
\DontPrintSemicolon
\setstretch{1.3}
\caption{}
\label{alg:greedy}
$U = X$ \tcp*{uncovered points}
$T = \emptyset$ \tcp*{picked centers forming an independent set in $\Mrel$}
\For{$i \leftarrow 1$ \KwTo $\rank(M)$}
{
determine  $t_i \in \argmax_{t\in X} \left\{ w(B(t,r)\cap U) \colon T \cup  \{t\} \in \mathcal{J} \right\}$\\
$T \leftarrow T \cup \{t_i\}$\\
$U\leftarrow U\setminus B(t_i,3r)$\\
}
\Return a set $R\in \mathcal{I}$ of distinct representatives for $T$
\end{algorithm}
 
\section{Analysis}\label{sec:analysis}

We now show that, given a Robust Matroid Center instance for which a solution of radius $r$ exists, \Cref{alg:greedy} efficiently finds a solution of radius $5r$, thus implying \Cref{thm:fromRToFiveR}.
As already discussed, \Cref{alg:greedy} is clearly efficient.
Hence, it remains to show that the set $R$ returned by \Cref{alg:greedy} fulfills
\begin{equation}\label{eq:toShowCoverageROk}
w(B(R,5r)) \geq m.
\end{equation}
To show~\eqref{eq:toShowCoverageROk}, let $T\subseteq X$ be the set computed by our algorithm (for which $R$ are distinct representatives) 
and we prove that $w(B(T,3r))$ is at least as large as the weight that can be covered by any independent set of centers within a radius of $r$.
In other words, let $F\in \mathcal{I}$ be a solution of radius $r$, i.e., $w(B(F,r)) \geq m$, which exists by assumption.
Moreover, without loss of generality, we assume that $F$ is a basis, hence, $|F|=\rho$, where $\rho$ is the rank of the matroid $M$.
In the following, we will show
\begin{equation}\label{eq:toShowCoverageTOk}
w(B(T,3r))\geq w(B(F,r)),
\end{equation}
which implies~\eqref{eq:toShowCoverageROk} because $w(B(F,r))\geq m$ and each point in $T$ has distance at most $2r$ from some point in $R$ (in particular, from its representative in $R$), thus implying $B(R,5r)\supseteq B(T,3r)$.

Note that we can also assume that $T$ is a basis in $\Mrel$.
Indeed, the only way how $T$ may not be a basis, is that the algorithm picks the same element multiple times.
Except for the first time, all other iterations where that element gets picked will not cover any additional points of $U$.
Hence, at such an iteration, any other point $t\in X$ with $T\cup \{t\}\in \mathcal{J}$ could have been picked.
Thus, if $T$ is not a basis, we can simply augment it to an arbitrary basis in $\Mrel$, which will not increase its coverage.
Note that this is also a legal outcome of our algorithm, as each element that was added to obtain a basis could have been picked instead of one that was picked multiple times.

To prove~\eqref{eq:toShowCoverageTOk}, we define an auxiliary matroid $\Mext$, which is an extension of $\Mrel$.
In $\Mext$, both $T$ and the centers $F$ (to be precise, a copy of it) are independent, and we show how $F$ can be transformed into $T$ one element at a time, while preserving coverage requirements in a well-defined way (to be specified in \Cref{lem:morphOursIntoF}).
$\Mext=(X\cup \overline{F},\Jext)$ is also a Rado matroid.
Its ground set is $X\cup \overline{F}$, where $\overline{F}$ is a copy of the points in $F$ (disjoint from $X$). 
To relate points in $\overline{F}$ to their counterparts in $F$, when referring to an element $\overline{f}\in \overline{F}$, we denote by $f$ the corresponding copy in $F$; analogously, for a set $\overline{H}\subseteq \overline{F}$, we denote by $H\subseteq X$ the set of points in $X$ that correspond to $\overline{H}$.
The independent sets $\Jext$ of $\Mext$ are given by
\begin{equation*}
\Jext \coloneqq \left\{
J\cup \overline{H} \ \middle \vert 
\begin{array}{c}   J\subseteq X,\, \overline{H}\subseteq\overline{F}  \text{ such that there exists a set } I\subseteq X\setminus H \text{ with } I\cup H\in \mathcal{I}\\
 \text{and a bijection } \phi\colon J\to I \text{ with } d(\phi(e),e)\leq 2r \;\forall e\in J
\end{array}
\right\}.
\end{equation*}

The family $\Jext$ indeed corresponds to the independent sets of a Rado matroid, as described by~\eqref{eq:radoIndep}, which follows by setting $X_e = B(e,2r)$ for $e\in X$ and $X_{\overline{f}} = 
\{f\}$ for $\overline{f}\in \overline{F}$.
In words, the matroid $\Mext$ extends the matroid $\Mrel$ with additional elements $\overline{F}$ and, for any $\overline{f}\in \overline{F}$, the representative of $\overline{f}$ in a set of distinct representatives is always $f$.
In particular, $\mathcal{J}\subseteq \Jext$. 
The way $\Mext$ is constructed allows us to capture both the solution given by $F$ and the points $T$ computed by our algorithm as independent sets in $\Mext$, where we represent $F$ by $\overline{F}$, and $T$ by itself.
In short, points of $X$ follow our relaxed notion in the sense that each point in $X$ can be replaced by another point within distance $2r$, whereas points in $\overline{F}$ do not have this flexibility.
Note that the rank of $\Mext$ is $\rank(\Mext)=\rho$.
This holds because the rank can clearly not be larger because $\Mext$ is a Rado matroid derived from $M$; moreover, any basis of $M$ is independent in $\Mext$.

The following statement shows how we can transform $F$ into $T$ as discussed above. (We recall that $t_1,\dots, t_{\rho}$ is the numbering of the elements of $T$ as constructed in \cref{alg:greedy}.)
\begin{lemma}
\label{lem:morphOursIntoF}
There is an ordering $f_1, \dots, f_\rho$ of the points of $F$ such that for all $j\in \{0,1,\dots,\rho\}$:
\begin{enumerate}[nosep,topsep=0.2em,itemsep=0.2em,label=(\roman*)]
\item\label{item:mixedBasis}
$\{t_1,\dots, t_j\} \cup \{\overline{f_{j+1}},\dots, \overline{f_{\rho}}\}$ is a basis of $\Mext$, and

\item\label{item:enoughCoverage}
$w\Big(B(\{t_1,\dots,t_j\},3r)\cup B(\{f_{j+1},\dots, f_{\rho}\},r)\Big) \geq w(B(F,r))$.
\end{enumerate}
\end{lemma}
First note that, by setting $j=\rho$, \Cref{lem:morphOursIntoF} implies~\eqref{eq:toShowCoverageTOk} as desired.
Hence, it remains to prove \Cref{lem:morphOursIntoF}.

\begin{proof}[Proof of \Cref{lem:morphOursIntoF}]
We will inductively determine the ordering $f_1,\dots, f_\rho$ of the elements in $F$ and prove the claimed statement.
The case $j=0$ is trivially true.
Hence, fix an index $j\in \{1,\dots, \rho\}$ and suppose that we have already determined elements $f_1,\dots, f_{j-1}$ of $F$ such that both~\ref{item:mixedBasis} and~\ref{item:enoughCoverage} hold for any index below $j$.
Note that \ref{item:mixedBasis} is a condition on the set $\{t_1,\dots, t_j\}\cup \overline{F}_{>j}$, where $\overline{F}_{>j}\coloneqq \overline{F}\setminus \{\overline{f_1},\dots,\overline{f_j}\}$; hence, it indeed only depends on the elements $f_1,\dots f_j$.
We use notation analogous to $\overline{F}_{>j}$ also for other numbered sets, like $F_{>j}\coloneqq F \setminus \{f_1,\dots, f_j\}$ and $T_{\leq j}=\{t_1,\dots, t_j\}$.
We distinguish two cases, depending on whether $B(t_j,2r) \cap F_{\geq j}$ is empty (see \Cref{fig:lemma}).
Note that to show~\ref{item:mixedBasis}, we only need to show that $T_{\leq j}\cup \overline{F}_{>j}$ is independent in $\Mext$, which immediately implies that the same set is a basis of $\Mext$ because $\rank(\Mext)=\rho$.

\begin{figure}[t]
\centering
\includegraphics[width=\textwidth]{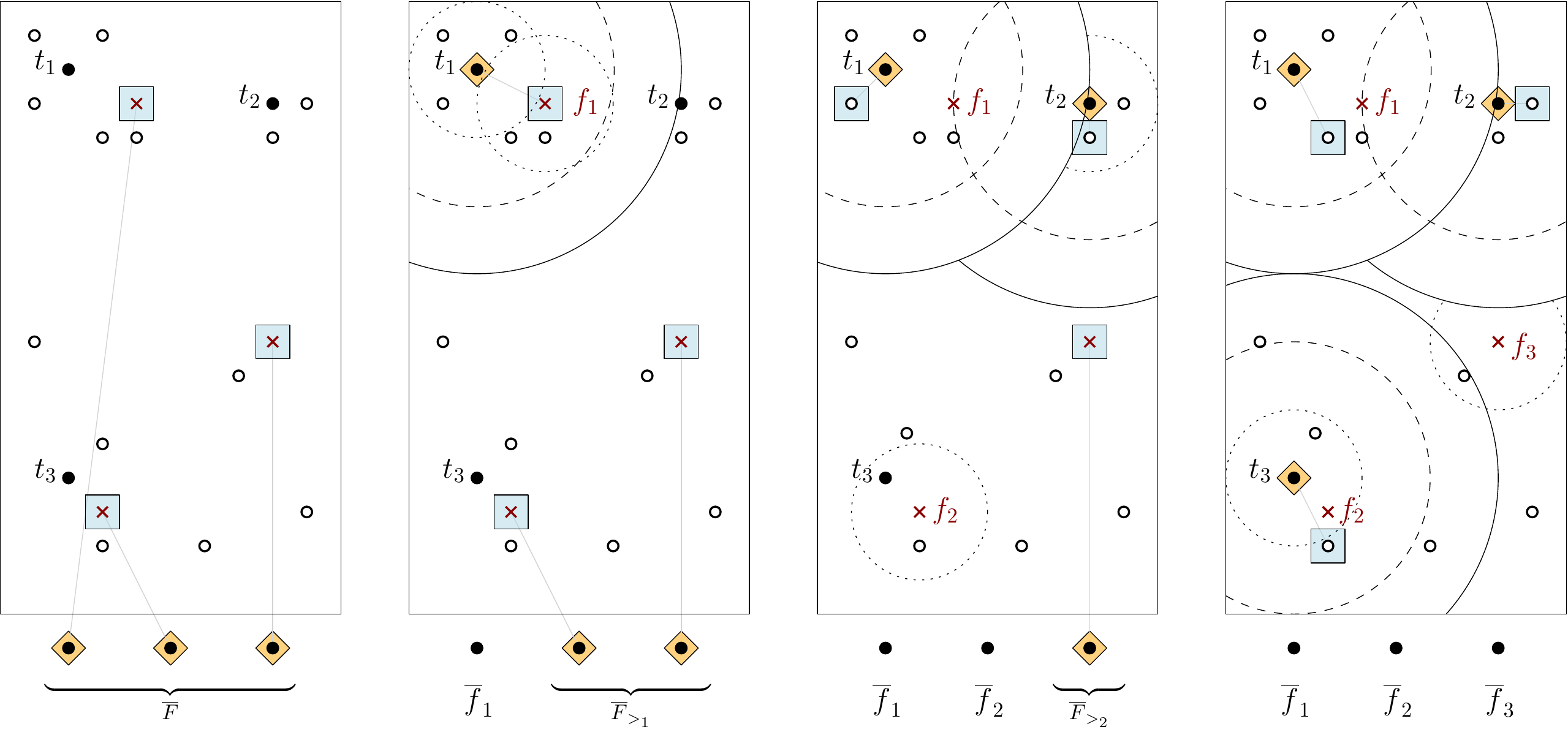}
\caption{
Illustration of a possible sequence of the exchange steps in the proof of \Cref{lem:morphOursIntoF}.
Starting from the basis $\overline{F}$ of $\Mext$ (and the corresponding solution $F$ of radius $r$) on the left, the inductive step is performed three times.
We arrive at the basis $T=\{t_1,t_2,t_3\}$ (and its representatives, which is the solution $R$ of radius $5r$ returned by the algorithm) on the right.
The points of $F$ are indicated by red crosses, and the basis maintained in the Rado matroid $\Mext$ is indicated by orange diamonds.
Each point in a basis is connected by a gray line to its representative, which is indicated by a blue square. 
Balls of radius $2r$ are indicated by dashed circles.
The parts of the balls of radius $r$ and radius $3r$ that are uncovered at the given iteration are indicated by dotted and solid lines, respectively.
Note that each exchange step may lead to changes in the representatives of the points in $T$ that are part of the basis; though, we do not need the set of representatives for intermediate iterations but only return the final one.
}
\label{fig:lemma}
\end{figure}

\medskip

\noindent \emph{Case 1:} $B(t_j,2r)\cap F_{\geq j} \neq \emptyset$.

In this case we choose $f_j\in B(t_j,2r)\cap F_{\geq j}$.
Property~\ref{item:mixedBasis} holds because a set of distinct representatives for $T_{\leq j-1}\cup \overline{F}_{\geq j}$ with respect to $\Mext$, which exists by the inductive hypothesis, is also a set of distinct representatives for $T_{\leq j}\cup \overline{F}_{> j}$.
Indeed, $f_j$, which had to be used as a representative for $\overline{f_j}$ for the set $T_{\leq j-1}\cup \overline{F}_{\geq j}$ is also a valid representative for $t_j$ because $f_j\in B(t_j,2r)$.

Moreover,~\ref{item:enoughCoverage} holds because $B(f_j,r)\subseteq B(t_j,3r)$.

\medskip

\noindent \emph{Case 2:} $B(t_j,2r)\cap F_{\geq j} = \emptyset$.

Adding $t_j$ to the basis $T_{\leq j-1} \cup \overline{F}_{\geq j}$ of $\Mext$ creates a unique circuit $C \subseteq T_{\leq j}\cup \overline{F}_{\geq j}$.
Note that $C\cap \overline{F}_{\geq j}\neq \emptyset$, because $T_{\leq j}$ is an independent set in $\Mrel$, which follows from $T_{\leq j}\subseteq T$ and the fact that $T\in \Jext$ because $T\in \mathcal{J}$ by construction and $\mathcal{J}\subseteq \Jext$.
Hence, let $\overline{f}\in C\cap \overline{F}_{\geq j}$ and we set $f_j = f$.
Property~\ref{item:mixedBasis} holds because $T_{\leq j}\cup \overline{F}_{>j}$ is obtained from $T_{\leq j}\cup \overline{F}_{\geq j}$, by removing an element from its unique circuit $C$; hence, it is circuit-free and thus independent in $\Mext$.

To show that property~\ref{item:enoughCoverage} holds, consider the $j$th iteration of the for-loop in \Cref{alg:greedy}.
Let
\begin{equation*}
U\coloneqq X \setminus B(T_{\leq j-1},3r)
\end{equation*}
be the state of $U$ at that iteration of the algorithm.
Note that at this iteration, the element $f_j$ was a candidate for $t_j$ in the maximization problem of \Cref{alg:greedy}.
Hence,
\begin{equation}\label{eq:tjAsGoodAsfj}
w(B(t_j,r)\cap U) \geq w(B(f_j,r)\cap U).
\end{equation}
By singling out the marginal contribution of the element $t_j$ on the left-hand side of the inequality~\ref{item:enoughCoverage} and further expanding, we obtain as desired:
\begin{align*}
w\left(B(T_{\leq j},3r)\cup B(F_{>j},r)\right)\hspace*{-2cm}&\\
&= w(B(T_{\leq j-1},3r) \cup B(F_{>j},r)) + w\left(B(t_j,3r)\setminus (B(T_{\leq j-1},3r) \cup B(F_{>j},r))\right)\\
&\geq w(B(T_{\leq j-1},3r) \cup B(F_{>j},r)) + w\left(B(t_j,r)\setminus (B(T_{\leq j-1},3r) \cup B(F_{>j},r))\right)\\
&= w(B(T_{\leq j-1},3r) \cup B(F_{>j},r)) + w\left(B(t_j,r)\setminus (B(T_{\leq j-1},3r)\right)\\
&= w(B(T_{\leq j-1},3r) \cup B(F_{>j},r)) + w\left(B(t_j,r) \cap U\right)\\
&\geq w(B(T_{\leq j-1},3r) \cup B(F_{>j},r)) + w\left(B(f_j,r) \cap U\right)\\
&\geq w(B(T_{\leq j-1},3r) \cup B(F_{\geq j},r))\\
&\geq w(B(F,r)),
\end{align*}
where the above relations hold due to the following.
The second equality holds by the disjointness assumption of the second case, i.e., $B(t_j,2r)\cap F_{\geq j}=\emptyset$, which implies $B(t_j,r) \cap B(F_{\geq j},r)=\emptyset$ and thus also $B(t_j,r) \cap B(F_{> j},r)=\emptyset$.
The third one follows from the definition of $U$. 
The second inequality is a consequence of~\eqref{eq:tjAsGoodAsfj}.
Finally, the last inequality holds due to the inductive hypothesis.
\end{proof}

\appendix

\section{Finding systems of distinct representatives in Rado matroids}\label{sec:distinctRepRado}

For completeness, we briefly repeat below how one can efficiently find systems of distinct representatives in Rado matroids.
\begin{lemma}
\label{lem:representatives}
Let $M=(X, \mathcal{I})$ be a matroid given by an independence oracle, let $Y$ be a finite ground set, and $X_y\subseteq X$ for $y\in Y$.
Let $\overline{M}=(Y,\mathcal{J})$ be the Rado matroid induced by $(M,\{X_y\}_{y\in Y})$.
Then, for any $J\subseteq Y$, one can efficiently determine whether $J\in \mathcal{J}$ and, if so, efficiently compute a system of distinct representatives for $J$.
\end{lemma}
\begin{proof}
Let $\widehat{\mathcal{M}} \coloneqq (X, \widehat{\mathcal{I}})$, where 
\[
\widehat{\mathcal{I}} \coloneqq \{ I \subseteq X \colon \text{ there is a bijection $\phi\colon J\to I$ with $\phi(y)\in X_y \;\forall y\in J$} \}. 
\]
$\widehat{\mathcal{M}}$ is well-known to be a matroid (see, e.g., \cite[Section~7.3]{W10}).
Independent sets in $\widehat{\mathcal{M}}$ fulfill the property of a system of distinct representatives for $J$ except for not needing to be independent in $M$.
Hence, a set $I\subseteq X$ is a system of distinct representatives for $J$ if and only if $I\in \widehat{\mathcal{I}}$ and $I\in \mathcal{I}$.
Thus, a system of distinct representatives for $J$ can be found using matroid intersection if it exists; otherwise $J\not\in \mathcal{J}$.
\end{proof}
 
\printbibliography

\end{document}